\documentclass[titlepage=no]{scrartcl}

\usepackage{stmaryrd}
\usepackage{amssymb}
\usepackage{amsthm}
\usepackage{amsmath}
\newtheorem{mydef}{Definition}
\newtheorem{myex}[mydef]{Example}
\newtheorem{mythm}[mydef]{Proposition}
\newtheorem{mylem}[mydef]{Proposition}
\newtheorem{myprop}[mydef]{Proposition}
\usepackage{proof}
\usepackage{bussproofs}
\usepackage{natbib}
\usepackage{doi}
\usepackage{xypic}
\usepackage{graphicx}

\usepackage{color}

\newcommand{\nata}{\mathit{Nat}}

\usepackage{enumitem}
\setlist{noitemsep}

\usepackage{scalerel,stackengine}
\stackMath
\newcommand\reallywidehat[1]{%
\savestack{\tmpbox}{\stretchto{%
  \scaleto{%
    \scalerel*[\widthof{\ensuremath{#1}}]{\kern-.6pt\bigwedge\kern-.6pt}%
    {\rule[-\textheight/2]{1ex}{\textheight}}%WIDTH-LIMITED BIG WEDGE
  }{\textheight}% 
}{0.5ex}}%
\stackon[1pt]{#1}{\tmpbox}%
}
\date{}
\begin{document}
\title{The far side of the cube}
\subtitle{An elementary introduction to game semantics}
\author{
  Dan R. Ghica~ \\
  University of Birmingham~ }
\maketitle

\begin{abstract}
Game-semantic models usually start from the core model of the prototypical language PCF, which is characterised by a range of combinatorial constraints on the shape of plays. Relaxing each such constraint usually corresponds to the introduction of a new language operation, a feature of game semantics commonly known as the \emph{Abramsky Cube}. In this presentation we relax all such combinatorial constraints, resulting in the most general game model, in which all the other game models live. This is perhaps the simplest set up in which to understand game semantics, so it should serve as a portal to the other, more complex, game models in the literature. It might also be interesting in its own right, as an extremal instance of the game-semantic paradigm. 
\end{abstract}

\section{Game semantics and definability}

\begin{quotation}
  \em{Thus we begin to develop a semantic taxonomy of constraints on strategies mirroring the presence or absence of various kinds of computational features.}
\end{quotation}
\begin{flushright}
  \citep{DBLP:journals/entcs/AbramskyM96}
\end{flushright}

A \emph{denotational semantics} models a programming language by translating it into a mathematical \emph{semantic domain}. This approach was pioneered by~\cite{ScottS71}, motivated by reasoning about program and compiler correctness. For mathematical and philosophical reasons it makes sense to define this translation function \emph{compositionally} from the structure of the language syntax. This translation function will not be, except in the most trivial cases, an \emph{injection} since the same semantic concept can have multiple syntactic representations. This is part of the challenge and attraction of understanding languages. On the other hand, also for mathematical and philosophical reasons, it is preferable if the translation function is a \emph{surjection}, meaning that every semantic concept has a syntactic representation, as pointed out by~\cite{lawvere1969adjointness}. Informally speaking, this simply means that there is no `\emph{junk}' in the semantics or, conversely, that there are no missing elements in the syntax. This property is called \emph{definability}.

There are other, more basic, requirements that a denotational semantics must meet. It must be \emph{sound}, meaning essentially that distinct syntactic entities are not wrongly identified (for example \emph{true} and \emph{false}, or 1 and 0), and it must be \emph{adequate} meaning that \emph{terminating} and \emph{non-terminating} programs are not mistakenly identified. 
Definability, soundness, and adequacy together establish that the translation is mathematically precise: it will translate all and only terms which are \emph{equivalent} in the syntax into \emph{equal} mathematical objects. This ideal situation is called \emph{full abstraction}. And, as many ideal situations, it turns out to be difficult to achieve, primarily because of failure of definability. This was first pointed out  by~\cite{DBLP:journals/tcs/Plotkin7} in the case of PCF, a simple yet surprisingly challenging functional language.

The failure of full abstraction indicates a mismatch between syntax and the semantic domain, which can be resolved in two ways. The first one is to enrich the syntax with the missing operations, a course of action taken in \emph{loc.\ cit.}, to wit, by adding a `parallel-or' operator. If we think of syntax as mere notation for semantic concepts, which we may hold as essential, this would seem the right course of action. There is however a second way to mend the gap, by removing from the semantics those objects that have no syntactic expression. This may seem a somewhat surprising concession to the preeminence of syntax, but it is more than that. Solving full abstraction, which means solving definability, is a challenging litmus test to the power of a semantic methodology. This difficult mathematical problem, in the context of any non-trivial (and not contrived) programming language, remained open for about two decades. \cite{DBLP:journals/iandc/HylandO00} give an excellent scholarly account of the quest for answering this question (Sec.~1.3, \emph{loc. cit.}). 

The solution to the problem of definability was brought about by \emph{game semantics}. For a tutorial introduction, history and overview of the subject the reader is referred to~\cite{samsong99,DBLP:conf/lics/Ghica09,DBLP:journals/ftpl/MurawskiT16}. Of particular interest to us is the~\cite{DBLP:journals/iandc/HylandO00} model, which along with~\cite{DBLP:journals/corr/AbramskyJM13}, is one of the original game-semantic models for PCF which achieves definability. It also introduces a style of game semantics, based on so-called \emph{pointer sequences}, which proved to be very successful because of its flexibility. Using this style of game semantics,~\cite{DBLP:journals/entcs/AbramskyM96} gave the first fully abstract model of \cite{reynolds81}'s intensely studied functional-imperative language \emph{Idealised Algol} (IA) --- \cite{o2013algol} collects these and other key papers on the semantics this language. 

The relation between PCF and IA is a very interesting one. Syntactically and operationally IA is a superset of PCF, to which it adds \emph{local state}. Despite this close connection, denotational models of IA differed significantly from those of PCF in terms of their mathematical structure. It was considered essential that the structure of the semantic domain mimics the structure of the store, something that was postulated by \cite{reynolds81} as one of the basic principles of the language: ``\emph{5. The language should obey a stack discipline, and its definition should make this discipline obvious.}'' Here `obvious' means that it should be part of the domain equations. This imperative led \cite{Oles_1983} to formulate an influential model based on functor categories. \cite{DBLP:journals/lisp/TennentG00} give a survey of the evolution of IA models.

However, since IA lives inside PCF it was likely that the full abstraction problems for the two are connected. And, indeed, the fully abstract model of IA followed shortly that of PCF. Even though the IA model is less celebrated than that for PCF, which solved a long-standing open problem of high stature, two of its features foreshadowed the coming success and dominance of the game-semantic methodology.

The first achievement of the IA model was rather technical. Part of the methodology of denotational semantics mandates that \emph{equivalent} syntactic phrases are mapped into \emph{equal} mathematical objects. This was achieved by~\cite{DBLP:journals/tcs/Milner77} using so-called \emph{term-model constructions}, starting from the syntax and applying quotients. But such models do not make semantic reasoning any easier. They are a form of sweeping under the rug. Game models for PCF are not syntactic, but they use a form of quotienting which was found by some to be objectionable, although the objections were largely expressed in the form of pub conversations rather than in formal publications. Moreover, as~\cite{DBLP:journals/tcs/Loader01} showed soon after, term equivalence for PCF is not decidable, so it was unlikely that the semantic domain of PCF was going to consist of neat mathematical objects. The IA game model put this debate to rest by providing a language interpretation in which no quotienting is required and thus eliminating a significant, if somewhat obscure, objection to games-based models.

The second achievement of the IA model was more subtle but at the same time more consequential. The model of IA was, in some sense, as close to the model of PCF as the syntax of IA is close to that of PCF. Both are interpreted in, essentially, the same semantic domain and the difference is a mere tweak. Even though the IA game model was foreshadowed by some earlier models, such as the object model of~\cite{DBLP:journals/lisp/Reddy96}, the similarity between it and that of PCF was striking, and it suggested that small tweaks to the game model can lead to models for diverse languages, starting from a common fundamental game model. The final paragraph of the paper (the version which appeared as a part of~\cite{o2013algol}) states that:
\begin{quotation}
  Another point for further investigation is suggested by the following diagram:
  \[
    \xymatrix{
      & \lnot I \land \lnot B & \\
      I\land \lnot B \ar[ur] & & \lnot I \land B\ar[ul]\\
      & I \land B\ar[ul]\ar[ur]&
    }
  \]
  Here $I$ denotes innocence and $B$ the bracketing condition [{\ldots}] and very successfully capture pure functional programming. As we have seen in the present paper, the category of knowing (but well bracketed) strategies captures IA. If we conversely retain innocence but weaken the bracketing condition then we get a model of PCF extended with non-local control operators. \em{Thus we begin to develop a semantic taxonomy of constraints on strategies mirroring the presence or absence of various kinds of computational features.}
\end{quotation}

The `\emph{innocence}' and `\emph{bracketing}' conditions mentioned above are the relatively small adjustments that the PCF game model requires in order to lead to full abstraction for other languages. The lattice of conditions above subsequently received new dimensions, and was dubbed a \emph{cube} by~\cite{samsong99}, which was then commonly referred to as `\emph{Abramsky's Cube}'. Exploring the various vertices of this (hyper)cube led to the development of many interesting and useful semantic models. Even though other methods such as \emph{trace semantics} also led to the development of fully abstract models for non-trivial languages~\citep{DBLP:conf/esop/JeffreyR05} it is fair to say that ultimately game semantics became the dominant paradigm, thanks in no small part to the guidance and inspiration provided by the Abramsky Cube. 

\subsection*{Beyond the cube, beyond definability}

The methodology of game semantics was naturally guided by its history, energised by the quest for PCF definability. This meant that the first game semantic model, that of PCF, was also in some sense the most highly constrained game semantic model. Other models are then derived by relaxing some of the constraints. This is perhaps paradoxical: Why does the simplest language (PCF) have the most complicated model? This is, again, because of definability. In a simple language relatively few semantic objects are definable. The constraints on the model are intended to rule out certain objects by deeming them to be `\emph{illegal}'. It is the remaining, legal, ones which are syntactically definable. As the language becomes richer, some of these semantic constraints can be relaxed. But this should lead to an obvious question: What if we relax all the constraints? Or, rather, what if we relax all the constraints that we can relax without making the model fall apart? What model lies at the top of the cube (or lattice, rather) of constraints. This is what the current essay will attempt to answer.

Since the game model on display here is simple, we aim for this to be a self-contained, accessible, and elementary introduction to game semantics. This presentation will be done in the style of \cite{DBLP:journals/entcs/GabbayG12} which will allow us to streamline some of the basic proofs of properties of game semantics. Arguably, the model we present here can be seen as the ur-model, at least for \emph{call-by-name} programming languages. Understanding it should give an easier access ramp to the rich, diverse, and mathematically sophisticated world of game semantics. 

\paragraph{Acknowledgments} Much of this material represents a tidying up of notes for courses taught at research summers schools, in particular the JetBrains Summer School, St.\ Petersburg (2016) and the Oregon Programming Languages Summer School (2018). 

\section{Game semantics, an interaction semantics}

\subsection{Arenas, plays, strategies}

The terminology of game semantics guides the intuition towards the realm of game theory. Indeed, there are methodological and especially historical connections between game semantics and game theory, but they are neither direct nor immediate. The games involved are rooted in logic and reach programming languages via the Curry-Howard nexus. They are not essential for a working understanding of game semantics as a model of programming languages, so we will not describe them here. But if we were to be pushed hard  to give a game-theoretic analogy, the ones to keep in mind are not the quantitative games of economics but rather last-player-wins games such as Nim. 

It is more helpful to think of game semantics as an interaction, or dialogue, between two agents, rather than a game. The dialogue is between a \emph{term} $t$, i.e.\ a piece of programming language code, and its \emph{context} $\mathcal C[-]$, i.e.\ the rest of the code. By placing the term in context we create an executable program $\mathcal C[t]$. During execution, certain interactions such as function calls and returns, or variable access, will happen. These are the interactions that are organised into a game model.

This interaction is asymmetric. One agent (P) represents the term and the other (O) represents an abstract and general context in which the term can operate.\footnote{The names stand for `\emph{Proponent}' and `\emph{Opponent}' even though there is nothing being proposed, and there is no opposition to it. The names are historical artefacts. We might as well call them `\emph{Popeye}' and `\emph{Olive}'. Same applies to `move', `play', and `strategy'.} The interaction consists of sequences of events called \emph{moves}, which can be seen as either calls, called \emph{questions}, or returns, called \emph{answers}. A sequence of events, with some extra structure to be discussed later, is called a \emph{play} and it corresponds to the sequence of interactions between the term and the context in one given program run. The set of all such possible plays, for all possible contexts, is called a \emph{strategy} and it gives the \emph{interpretation} of the term. The strategy of a term can be constructed inductively on its syntax, from basic strategies for the atomic elements and a suitable notion of composition to be discussed later.

Before we proceed, a caveat. The structure of a game semantics is dictated by the evaluation strategy of the language and its type structure. Call-by-name games are quite differently structured than call-by-value games. Hereby we shall assume a call-by-name evaluation strategy and simple type discipline of base types and functions. The reason is didactic, as these games are easier to present. Having understood game semantics in this simple setting, understanding other more complex setups should be easier. 

Let us consider a most trivial example, the term consisting of the constant 0. The way this term can interact with any context is via two moves: a question ($q$) corresponding to the event interrogating the term, and an answer ($0$) corresponding to the term communicating its value to the context. The sequence $q\cdot 0$ is the only possible play, therefore the strategy corresponding to the set of plays $\{q\cdot 0\}$ is the interpretation of the term 0.

\begin{center}
  \includegraphics[scale=1.2]{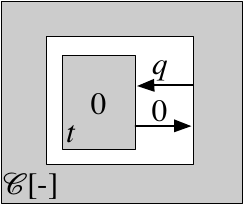}
\end{center}

This behaviour is `at the interface' in the sense that any other term evaluating to 0, such as $7-7$ or $4\times 2 -8$ would exhibit the same interaction. 

Let us consider a slightly less trivial example, the identity function over natural numbers $\lambda x.x : nat \rightarrow nat$.  The context can call this function, but also the function will enquire about the value of its argument $x$. Lets call these questions $q$ and $q'$. The context can answer to $q'$ with some value $n$ and the term will answer to $q$ with the same value $n$. Even though the answers carry the same value they are different moves, so we will write $n$ and $n'$ to distinguish them, where then prime is a syntactic tag. Plays in the strategy interpreting the identity over natural numbers have shape $q\cdot q'\cdot n'\cdot n$. Equivalent terms such as $(\lambda x.x)(\lambda x.x)$ exhibit identical interactions.

\begin{center}
  \includegraphics[scale=1.2]{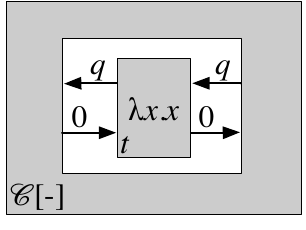}
\end{center}

Let us now define the concepts more rigorously.

\begin{mydef}[Arena]
  An \emph{arena} is a tuple $\langle M, Q, I, O, {\vdash}\rangle$ where
  \begin{itemize}
  \item $M$ is a set of \emph{moves}.
  \item $Q\subseteq M$ is a set of \emph{questions}; $A=M\setminus Q$ is the set of \emph{answers}.
  \item $O\subseteq M$ is a set of \emph{O-moves}; $P=M\setminus O$ is the set of \emph{P-moves}.
  \item $I\subseteq Q\cap O$ is a set of \emph{initial moves};
  \item ${\vdash}\subseteq M\times M$ is an \emph{enabling} relation such that if $m\vdash n $ then 
    \begin{itemize}
    \item[(e1)] $m\in Q$
    \item[(e2)] $m\in O$ if and only if $n \in P$
    \item[(e3)] $n \not\in I$.
    \end{itemize}
  \end{itemize}
\end{mydef}
An arena represents the set of moves associated with a type, along with the structure discussed above (questions, answers, O, P). Additionally, the arena introduces the concept of \emph{enabling} relation, which records the fact that certain moves are causally related to other moves. Enabling requires certain preliminary conditions:
\begin{itemize}
\item[(e1)] Only questions can enable other moves, which could be interpreted by the slogan `\emph{all computations happen because a function call}'. 
\item[(e2)] $P$-moves enable $O$-moves and \emph{vice versa}. Game semantics records behaviour \emph{at the interface} so any action from the context enables an action of the term, and the other way around. 
\item[(e3)] There is a special class of O-questions called \emph{initial moves}. These are the moves that are allowed to kick off an interaction, so do not need to be enabled. 
\end{itemize}
The informal discussion above can be made more rigorous now. 
\begin{myex}
  Let $\mathbf 1=\{ \star \}$. 
  The arena of natural numbers is $N =\langle \mathbf 1\uplus\mathbb N, \mathbf 1, \mathbf 1, \mathbf 1, \mathbf 1 \times \mathbb N \rangle$. 
\end{myex}
More complex arenas can be created using product $\times$ and arrow $\Rightarrow$ constructs. Let   
\begin{align*}
  inl &: M_A\rightarrow M_A+M_B\\
  inr &: M_B\rightarrow M_A+M_B
\end{align*}
where $+$ is the co-product of the two sets of moves. We lift the notation to relations, $R+R'\subseteq (A+A')\times(B+B')$:
\begin{align*}
  inl(R) &= \{(inl(m), inl(n)) \mid (m, n)\in R\}\\
  inr(R') &= \{(inr(m), inr(n)) \mid (m, n)\in R'\}.
\end{align*}

\begin{mydef}[Arena product and arrow]
  Given arenas $A=\langle M_A, Q_A, O_A, I_A, {\vdash}_A\rangle$ and $B=\langle M_B, Q_B, O_B, I_B, {\vdash}_B\rangle$ we construct the \emph{product arena} as
  \[
  A\times B =\bigl\langle
  M_A+ M_B,
  Q_A +Q_B,
  O_A +O_B,
  I_A +I_B,
  {\vdash}_A+{\vdash}_B
  \bigr\rangle
  \]
  and the \emph{arrow arena} as
  \begin{align*}
    A\Rightarrow B = \bigl\langle
    M_A+M_B,
    Q_A+Q_B,
    P_A+O_B,
    inr(I_B),
    {\vdash}_A+{\vdash}_B\cup inr(I_B)\times inl(I_A)
    \bigr\rangle.
  \end{align*}
\end{mydef}
If we visualise the two arenas as DAGs, with the initial moves as sources and with the enabling relation defining the edges, then the product arena is the disjoint union of the two DAGs and the arrow arena is the grafting of the $A$ arena at the roots of the $B$ arena, but with the O-P polarities reversed. 

Since arenas will be used to interpret types we can anticipate by noting that
\begin{mythm}[Currying]
For any arenas $A,B,C$ the arenas $A\times B\Rightarrow C$ and $A\Rightarrow B\Rightarrow C$ are isomorphic. 
\end{mythm}
\begin{proof}
  Both arena constructions correspond to the DAG below.
  \begin{center}
  \includegraphics[scale=1.2]{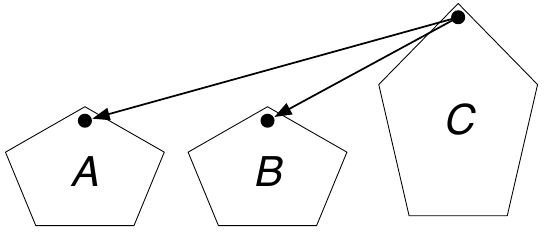}
  \end{center}
  The isomorphism is a node-relabeling isomorphism induced by the associativity isomorphism of the co-product. 
\end{proof}
We also note that
\begin{mythm}[Unit]
The arena $I = \langle \emptyset,\emptyset, \emptyset,\emptyset \rangle$ is a unit for product, i.e. for any arena $A$, $A\times I$, $I\times A$, $I\Rightarrow A$ are isomorphic to $A$. 
\end{mythm}
The isomorphism is a re-tagging of moves. 

\begin{myex}
  We talked earlier about the arena for the type $nat\rightarrow nat$. Let $inl(m)=m''$ and $inr(m)=m'$, where $'$ and $''$ are syntactic tags. The arena $\mathit{Nat}\Rightarrow \mathit{Nat}$ is represented by the DAG below
  \begin{center}
    \includegraphics[scale=1.2]{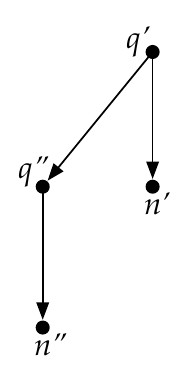}
  \end{center}
\end{myex}
We already mentioned that for the identity all plays have the shape $q'q''n''n'$. We note that in this particular play all move occurrences are preceded by an enabling move. The move corresponding to the term returning a value $n'$ can happen because the context initiated a play $q'$. The term can ask for the value of the argument $q''$ also because $q$ has happened earlier. Each move occurrence is \emph{justified} by an enabling move, according to $\vdash$, occurring earlier. The enabling relation defines the causal skeleton of the play. 

Let us further consider another term in the same arena $\lambda x.x+x$. How does this term interact with its context?
\begin{enumerate}
\item the context initiates the computation
\item the term asks for the value of $x$
\item the contest returns some $m$
\item the term asks again for the value of $x$
\item the context returns some $n$
\item the term returns to the context $m+n$.
\end{enumerate}
The reader familiar with call-by-value may be rather confused as to why the context returns first an $m$ and then an $n$. This is because in call-by-name arguments are thunks. In some languages the thunks may contain side-effects, which means that repeat evaluations may yield different values.

\newcommand{\cobnd}[1]{\langle #1 \rangle}

Looking at the arena, this interaction corresponds to the play $q'q''m''q''n''p'$, where $p=m+n$. The causal structure of this play is a little confusing. There are two occurrences of $q''$, the first one preceding both $m''$ and $n''$ and the second one only $n''$. It should be that the first occurrence of $q''$ enables $m''$ and the second enables $n''$, to reflect the proper call-and-return we might expect in a programming language. In order to do that the plays will be further instrumented with names called \emph{pointers}. Each question has a symbolic address, and is paired with the address of some other move, called \emph{enabling move}.. The fully instrumented play is called a \emph{pointer sequence}.

\newcommand{\ssin}{\,{\sqsubset}\hspace{-1.8ex}{-}\,}

Let us use $\epsilon$ for the empty sequence, $\cdot$ for sequence concatenation, $\sqsubseteq$ for sequence prefix ($u\sqsubseteq u\cdot v$) and $\ssin$ for the sub-sequence relation ($u\ssin v\cdot u\cdot w$). If unambiguous we may represent concatenation simply as juxtaposition ($uv$). Let $\mathbb A$ be a set of `\emph{names}'. 

\begin{mydef}[Pointer sequence]
  Given a set $M$, a \emph{pointer sequence} $p\in J_M$ is a sequence $p\in (M\times \mathbb A\times \mathbb A)^*$ such that for all $q\cdot (m,a,b)\sqsubseteq p$, for all $(m',a',b')\ssin q$, $b\neq a'$ and $b\neq b'$. 
\end{mydef}
We write these triples as $ma\cobnd b$.  The addresses $a,b,c,\ldots\in\mathbb A$ are just names, and the notation $\cobnd b$ means that the name $b$ is `\emph{fresh}', i.e.\ not used earlier in the sequence. We sometimes write this condition as $b\,{\#}\,q$. In general we will employ the Barendregt name convention that if two names $a,b$ are denoted by distinct variables they are distinct $a\neq b$. Answers never justify (in these games) so we may write their unused name as $\_$ or we may omit the whole $\cobnd{\_}$ component altogether.  In a pointer sequence, by a \emph{move occurrence} we mean the move along with the justifier $a$ and, if it is the case, the pointer name $\cobnd b$, taken as a whole.  We write $J_A$ for the set of pointer sequences over the moves of an arena $A$.

Going back to the earlier example, the interaction corresponds to the following pointer sequence:

\[
q'a \cobnd b\cdot q''b\cobnd c\cdot m''c\cdot q''b\cobnd d\cdot n''d\cdot p'b,
\]
noting that  $a$ is the only name without a previous binder, and is used by the initial question. If we were to represent the pointers graphically, the sequence above would be:
\begin{center}
  \includegraphics[scale=1.2]{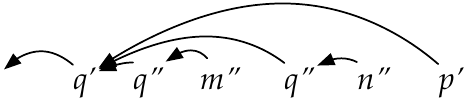}
\end{center}
In general, if we draw a pointer sequence we will omit the `dangling' pointers from the diagram.

The \emph{pointer sequence} represents not only the actions, that is the calls and returns, but also what calls correspond to what returns and even what calls are caused by other calls. In a sequential language, for terms up to order three the pointers can be actually uniquely reconstructed from the sequence itself. Otherwise the justification pointers are necessary.

\begin{mydef}[Play]
  A pointer sequence $p\in J_A$ is said to be a \emph{play} when:
  \begin{itemize}
  \item for any $p'\cdot ma\cobnd c \sqsubseteq p$, $p'\neq\epsilon$, there exists  $q\in Q_A,b\in\mathbb A$ such that $q b\cobnd a\ssin p'$ and $q\vdash_A m$.
  \item if $qa\cobnd b\sqsubseteq p$ then $q\in I_A$. 
  \end{itemize}
\end{mydef}
Above, it is implicit that $m\in M_A$, and $a,b,c\in\mathbb A$.

\newcommand{\Play}{\mathbf{P}}

We write the set of plays of an arena $A$ as $\Play_A$. They represent computations which are \emph{causally sensible}, so that pointers are consistent with the enabling relation. The behaviour of a function that would return or ask for its argument without being itself called is, for example, not causally sensible and its corresponding interactions are thus not plays.

Let $\pi:\mathbb A \rightarrow \mathbb A$ be bijections representing \emph{name permutations}, and define \emph{renaming actions} of a name permutation on a pointer sequence over arena $A$ as
\begin{align*}
  \pi\bullet \epsilon = \epsilon,\quad
  \pi\bullet (p\cdot ma\cobnd b) = (\pi\bullet p)\cdot (m\,\pi(a)\,\cobnd{\pi(b)}) 
\end{align*}

\begin{mylem}
  If $p\in \Play_A$ then $\pi \bullet p \in \Play_A$, for any bijection $\pi:\mathbb A\rightarrow \mathbb A$.
\end{mylem}

The proofs are elementary. If they are also straightforward we will leave them as exercise. 

A \emph{strategy} in an arena $A$ is any set of plays which is prefix closed, closed under choices of pointer names (\emph{equivariant}), and closed over $O$-moves.

\begin{mydef}[Strategy]
  A strategy over an arena $\sigma : A$ is a set of plays such that for any $p\in\sigma$ the following properties hold:
  \begin{description}
  \item[prefix-closed] If $p' \sqsubseteq p$ then $p'\in\sigma$
  \item[O-closed] If $ p\cdot m\in \Play_A$ for some $m\in O_A$ then $p\cdot m\in \sigma$
  \item[equivariance] For any permutation $\pi$, $\pi\bullet p\in \sigma$.
  \end{description}
\end{mydef}

The conditions above have intuitive explanations.
Prefix-closure is a natural condition on trace semantics, going back to~\cite{hoare1978communicating}'s model of CSP and beyond, to common encodings of trees in set theory. It has a clear causal motivation in the sense that a trace semantics is a history of behaviour, and any history must be prefix-closed.

The $O$-closure condition reflects the fact that a term has no control over which one of a range of possible next moves the context might choose to play. Finally, pointer equivariance is akin to an alpha-equivalence on plays, motivated by the fact that pointer names are not observable, so the choice of particular names in a play is immaterial. Name, equivariance and related concepts and reasoning principles are comprehensively studied by \cite{pitts2013nominal}.

In this most general setting names introduced by moves are required to be fresh. In more constrained settings it can be determined that a name is no longer to be used, because any use of that name would violate the constraints. In fact most game models in the literature have this property, with the model presented here being an exception. If a name is no longer usable than it is possible to introduce a notion of `\emph{scope}' for that name, raising the possibility of name reuse. \cite{DBLP:conf/csl/GabbayGP15} give a formulation of game models where pointer names are scoped. 

Note that it is equivalently possible to present strategies as a next-move function from a play to a $P$-move (or set of $P$-moves) along with the justification infrastructure, indicating the next move in the play. If $\mathcal P$ is the power-set, then
\begin{align*}
\hat\sigma : \Play_A \rightarrow \mathcal P(P_A\times {\mathbb {A}}^2), 
\end{align*}
such that for all $p\in\Play_A$, and $(m,a,b)\in\hat\sigma(p)$, $p\cdot ma\cobnd b\in\Play_A$.

\begin{mylem}
  Let $\hat\sigma : \Play_A \rightarrow \mathcal P(P_A\times {\mathbb {A}}^2)$ be a next-move function, and let $\sigma$ be the smallest set such as
  \begin{description}
    \item[Empty:] $\epsilon \in \sigma $
    \item[P-move:] $\text{if }p\in\sigma \text{ and } m a \cobnd b\in \hat\sigma(p) \text{ then } p\cdot m a \cobnd b \in \sigma$
    \item[O-move:] $\text{if }p\in\sigma \text{ and } p\cdot m a \cobnd b\in \Play_A^O\text{ then } p\cdot m a \cobnd b \in \sigma$.
  \end{description}
  The set $\sigma$ is a strategy. 
\end{mylem}
\begin{proof}
The set $\sigma$ is prefix-closed by construction. Adding all O-moves wherever legal ensures O-completeness. Equivariance holds from general principles. 
\end{proof}
When we specify a next-move function, if the result is a singleton set $\{ma \cobnd b\}$ we may simply write $ma \cobnd b$, and if $\hat\sigma(p)=\emptyset$ we may omit that case from the definition. 

We will sometimes define strategy directly and some  other times via the next-move function, whichever is more convenient.
  
\newcommand{\strat}{\mathrm{strat}}
\begin{mydef}
Given a set of plays over an arena $A$, $\sigma\subseteq P_A$, let us write $\strat(\sigma)$ for the least strategy including $\sigma$.
\end{mydef}

We will sometimes abuse the notation above by applying it to a set of sequences of moves, in the case that the pointer structure can be unambiguously reconstructed.

\begin{myex}
  In arena $\nata$,
  \[
  \sigma_0 = \strat(q\cdot 0)=
  \strat(qa\cobnd b\cdot 0b)=
  \{\epsilon, qa\cobnd b, qa\cobnd b\cdot 0b \mid a,b\in\mathbb A\}.
  \]
  The next-move function is
  \[
  \hat\sigma_0(qa\cobnd b)=0a.
  \]
  
\end{myex}

\subsection{Examples of strategies}

If we consider programming languages in the style of~\cite{DBLP:journals/cacm/Landin66}, i.e.\  the simply-typed lambda calculus with additional operations, such additional operations can be defined in the unrestricted game model. Let us consider several examples. 

\subsubsection{Arithmetic}

Any arithmetic operator $\circledast : nat\rightarrow nat\rightarrow nat$ is interpreted by a strategy over the arena $\nata\Rightarrow\nata\Rightarrow\nata$. Let us tag the moves of the first $\nata$ arena with ${-}_1$, the moves of the second with ${-}_2$ and leave the third un-tagged (the trivial tag). Then the interpretation of the operator is in most cases
\[
\sigma_{\circledast}=
\strat\left(\{
qq_1m_1q_2n_2p \mid m, n, p\in\mathbb N \land p = m \circledast n
\}\right)
\]
Note that in the case of division, or any other operation with undefined values, the strategy must include those cases explicitly:
\[
\sigma_{\div}=
\strat\left(\{
qq_1m_1q_2n_np \mid m, n\neq 0, p\in\mathbb N\land p = m \div n
\} \cup
\{qq_1m_1q_20_2 \mid m\in\mathbb N
\right)
\]
Following the $0$ $O$-answer, there is no way $P$ can continue, i.e.\ $\hat\sigma_{\div}(qq_1m_1q_20_2)=\emptyset$.

From this point of view sequencing can be seen as a degenerate operator which evaluates then forgets the first argument, then evaluates and return the second
\[
\sigma_{seq}=
\strat\left(\{
qq_1m_1q_2n_2n \mid m, n\in\mathbb N
\}\right)
\]
Of course, sequencing commonly involves commands $com$ which are degenerate, single-value, data types which are constructed just like the natural numbers but using a singleton set instead of $\mathbb N$. 

The flexibility of the strategic approach also gives an easy interpretation to shortcut (lazy) arithmetic operations:
\[
\sigma_{\times}=
\strat\left(\{
qq_1m_1q_2n_2p \mid m\neq 0, n, p\in\mathbb N\land p = m \times n
\} \cup
\{qq_10_10\}
\right)
\]

This comes in handy when implementing an if-then-else operator (over natural numbers), in arena $\mathit{ Bool}\Rightarrow\nata\Rightarrow\nata\Rightarrow\nata$, where $\mathit{Bool}$ is the arena of Booleans, with $M_{\mathit{Bool}}=\mathbb B=\{\mathit{tt},\mathit{ff}\}$:

\[
\sigma_{\mathit{if}}=
\strat\left(\{
qq_1\mathit{tt}_1q_2n_2n \mid n\in\mathbb N
\} \cup
\{
qq_1\mathit{ff}_1q_3n_3n \mid n\in\mathbb N
\}
\right)
\]

These strategies are found in the original PCF model of \cite{DBLP:journals/iandc/HylandO00}.

\subsubsection{Non-determinism}

A non-deterministic Boolean choice operator $\mathit{chooseb}:\mathbb B$ is interpreted by the strategy
\[
\sigma_{chooseb}=\strat(\{q\cdot tt, q\cdot \mathit{ff}\})
\]
where two $P$-answers are allowed. This can be extended to probabilistic choice by adding a probability distribution over the strategy. 

Note that the flexibility of the strategic approach can allow the definition of computationally problematic operations such as unbounded non-determinism $\mathit{choosen}:\mathbb N$,
\[
\sigma_{choosen}=\strat(\{qn \mid n \in \mathbb N\})
\]

Nondeterministic and probabilistic game semantics have been studied by~\cite{DBLP:conf/lics/HarmerM99} and~\cite{DBLP:conf/lics/DanosH00}, respectively. In terms of the Abramsky Cube, these games lead to definability via the relaxation of the determinism condition, which means that the strategy function can result in more than one possible move. By contrast, deterministic language strategies respond with at most one move for any given play. 

\subsubsection{State}

In order to model state we first need to find an appropriate arena to model assignable variables. In the context of call-by-name it is particularly easy to model \emph{local} (bloc) variables ($new\ x\ in\ t$, where $x$ is the variable name and $t$ the term representing the variable block). It turns out that $new$ needs not be a term-former but it can be simply a higher order language constant $new : (var\rightarrow T)\rightarrow T$ where $T$ is some language ground type.

The type of variables $var$ can be deconstructed following an `object oriented' approach initially proposed by~\cite{reynolds81}. A variable must be readable $der : var \rightarrow nat$ and assignable $\mathit{asg}:var\rightarrow nat\rightarrow nat$. Since no other operations are applicable, we can simply define $var = nat \times (nat \rightarrow nat)$ which means $der = proj_1$, $\mathit{asg}=proj_2$, so assignment behaves like in C, returning the assigned value. We will see later how projections are uniformly interpreted by strategies.

What is interesting is the interpretation of the $new$ operation in arena $(\mathit{Var}\Rightarrow\mathit{Unit})\Rightarrow \mathit{Unit}'$. With the decomposition above in mind and for the sake of clarity we call the moves in $\mathit{Var}$ as follows:
\begin{description}
\item[Read request]: $rd$
\item[Value read]: $val(n)$ for $n\in\nata$
\item[Write request]: $wr(n)$ for $n\in\nata$
\item[Value written]: $ok(n)$ for $n\in\nata$.
\end{description}

We define this strategy using the next-move function. The strategy will include copy-cat moves between $\mathit{Unit}$ and $\mathit{Unit'}$ along with stateful moves.

If a read ($rd$) move is played by O then P will respond with $val(n)$ where $n$ is the last move it played before, which can be a read value $val(n)$ or a written acknowledgment $ok(n)$.

\[
\hat\sigma_{new}(p\cdot m(n) a \cdot rd\,b\cobnd c) = val(n)c,
\quad
m\in\{val, ok\}.
\]
If a write $wr(n)$ move is played by O then P will always acknowledge it with $ok(n)$. 
\[
\hat\sigma_{new}(p\cdot wr\,a\cobnd b) = val(n)b
\]

These strategies were introduced by \cite{DBLP:journals/entcs/AbramskyM96} in the game model of IA. On the Abramsky Cube these strategies relax the `\emph{innocence}' condition of the strategy function which states that for any play $p$ there is a smaller play computed from it, called `\emph{the view}' $\lceil p\rceil$ such that $\hat\sigma(p)=\hat\sigma(\lceil p\rceil)$. In other words the term has a `\emph{restricted memory}' of the play in choosing the next move. This subtle condition has been studied quite by \cite{DBLP:conf/csl/DanosH01,DBLP:conf/fsttcs/HarmerL06,DBLP:conf/lics/HarmerHM07}.

\subsubsection{Control}\label{sec:control}

If-then-else is a very simple control operator, but more complex ones can be defined. The family of control operators is large, so let us look at a simple one, $catch:(com_1\rightarrow nat_2)\rightarrow option\ nat$ where the type $option\ nat$ is interpreted in an arena constructed just like $\nata$ but using $\mathbb N+1$ instead of $\mathbb N$. The extra value indicates an error result ($\baro$). Just like in the case of state, the construct $catch(\lambda x.t)$ can be sugared as $\mathit{escape}\ x\ in\ t$. If $x$ is used in $t$ then the enclosing $catch$ returns immediately with $\baro$, otherwise it returns whatever $t$ returns.

The strategy is
\begin{align*}
\sigma_{catch} =\strat\left(  \{ qq_2n_2n \mid n\in \mathbb N \} \right)
\cup  \{ qq_2q_1\baro\} 
)
\end{align*}

Game semantics for languages with control have been initially studied by~\cite{DBLP:conf/lics/Laird97}. In the Cube, these strategies relax the `\emph{bracketing}' constraint of the PCF model, which requires questions and answers to nest like well-matched brackets. 

\subsubsection{Concurrency}

As our final example we will consider \emph{parallel composition} of commands, $par:com_1\rightarrow com_2\rightarrow com$. In the case of this strategy, which represents a function which executes its arguments \emph{asynchronously} all interleavings of the two argument executions are acceptable:
\begin{align*}
  \sigma_{par} = \strat\left(q\cdot (q_1a_1\mid q_2a_2) \cdot a\right), 
\end{align*}
where $p\mid q$ is the set of all interleavings of two sequences.

The constraint relaxed by this strategy is the \emph{alternation} of $O$/$P$ moves, and was first studied by~\cite{DBLP:journals/apal/GhicaM08}. 

\subsection{Composing strategies}

In the previous section we looked at strategies interpreting selected language constants. In order to construct an interpretation of terms, denotationally, strategies need to \emph{compose}.

The intuition of composing a strategy $\sigma:A\Rightarrow B$ with a strategy $\tau:B\Rightarrow C$ is to use arena $B$ as an interface on which in a first instance $\sigma$ and $\tau$ will synchronise their moves.
\begin{center}
  \includegraphics[scale=1.2]{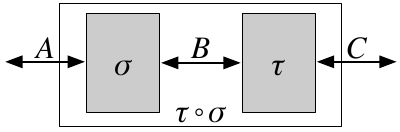}
\end{center}
After that, the moves of $B$ will be hidden, resulting in a strategy $\sigma;\tau:A\Rightarrow C$. In order to preserve proper justification of plays all pointers that factor through hidden moves will be `extended' so that the hiding of the move will not leave them dangling. In order to define composition some ancillary operations are required.

\newcommand{\delm}{\downharpoonright}

The first one is deleting moves while extending the justification pointers over the deleted moves, in order to preserve justification~\citep{DBLP:journals/entcs/GabbayG12}.
\begin{mydef}[Deletion]
  Let $X\subseteq M$ be sets. For a pointer sequence $p\in J_M$ we define \emph{deletion} inductively as follows, where we take $(p', \pi)=p\delm X$:
  \begin{align*}
    \epsilon\delm X &= (\epsilon, id) \\
    (p\cdot ma\cobnd b) \delm X &= (p'\cdot m\cdot \pi(a)\cobnd b, \pi) & \text{if $m\not\in X$}\\
    (p\cdot ma\cobnd b) \delm X &= (p', (\pi\mid b\mapsto \pi(a))) & \text{if $m\in X$}
  \end{align*}
\end{mydef}
The result of a deletion is a pointer sequence along with a function $\pi:\mathbb A\rightarrow \mathbb A$ which represents the chain of pointers associated with deleted moves. Informally, by $p\delm X$ we will understand the first projection applied to the resulting pair. Since deletion only removes names, it is immediate that $p\delm X$ is a well-formed pointer sequence.
\begin{mylem}
If $X\subseteq M$ and  $p\in J_M$ then $p\delm X\in J_{M\setminus X}$.
\end{mylem}

For example, the removal of the grayed-out moves in the diagrammatic representation of the play below results in a sequence with reassigned pointers:
\begin{center}
  \includegraphics[scale=1.2]{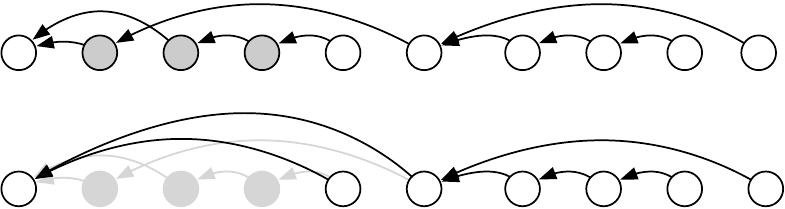}
\end{center}
Note that in general the deletion of an arbitrary set of moves from a play does not result in another play, since the pointers may be reassigned in a way that is not consistent with the enabling relation.  There is however an important special situation:
\begin{mylem}\label{lem:delplay}
  Given arenas $A,B$ if $p\in \Play_{A\Rightarrow B}$ then $p\delm A \in \Play_{B}$.
\end{mylem}
\begin{proof}
  Since the enabling relation is a DAG, no $A$-move enables any $B$-move. In a play $p$ thus there can be no pointer from an $A$-move occurrence to a $B$-move occurrence. Therefore when the $A$-move occurrences are deleted no pointer reassignment is required, so the result is a $B$-play. 
\end{proof}

The second operation is the selection of `hereditary' sub-plays, i.e. all the moves that can be reached from an initial set of moves following the justification pointers~\citep{DBLP:journals/entcs/GabbayG12}.

\newcommand{\selm}{\upharpoonright}

\begin{mydef}[Hereditary justification]
  Let $X\subseteq M$ be sets. For a pointer sequence $p\in J_M$ we define the \emph{hereditarily justified} sequence $p\selm X$ recursively as below, where we take $p\selm X=(p',X')$:
  \begin{align*}
    \epsilon\selm X &= (\epsilon, X) \\
    (p\cdot ma\cobnd b) \selm X &= (p'\cdot m a\cobnd b, X\cup\{b\}) & \text{if $a\in X$}\\
    (p\cdot ma\cobnd b) \selm X &= (p', X) & \text{if $a\not\in X$}\\
  \end{align*}
\end{mydef}

The result of a hereditary justification is a pointer sequence along with a set of names $X\subseteq\mathbb A$ which represents the addresses of selected questions. Informally, by $p\selm X$ we will understand the first projection applied to the resulting pair. Since hereditary justification only removes names, it is immediate that $p\selm X$ is a well-formed pointer sequence.
\begin{mylem}
If $X\subseteq M$ and  $p\in J_M$ then $p\selm X\in J_{X}$.
\end{mylem}

For example, the hereditary justification of the grey move in the diagrammatic representation of the play below results in the sequence below:
\begin{center}
  \includegraphics[scale=1.2]{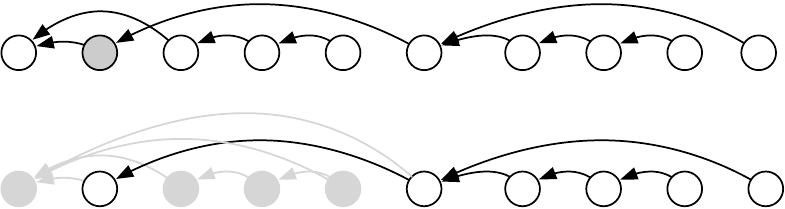}
\end{center}

In general the hereditary justification of an arbitrary set of moves from a play does not result in another play, since it may result in a sequence that does not start with an initial move. The special situation is:

\begin{mylem}\label{lem:selplay}
  Given arenas $A,B$ if $p\in \Play_{A\Rightarrow B}$, with $ma\cobnd b\ssin p$, $m\in I_A$ then $p\selm \{b\} \in \Play_A$.
\end{mylem}
\begin{proof}
  Since the enabling relation is a DAG, no $A$-move enables any $B$-move. In a play $p$ thus there can be no pointer from an $A$-move occurrence to a $B$-move occurrence. Therefore when the hereditarily justified sequence from an $A$-initial move is selected it will result in an $A$-play. 
\end{proof}

Props.~\ref{lem:delplay} and~\ref{lem:selplay} are technically important, and they are consistent with the intuitive model of computation we relied on. What they say is that in a play corresponding to a function of type $A\Rightarrow B$, the sequences associated just with the argument~$A$ or just with the result~$B$ are also plays. In other words, when the behaviour of a function is causally sensible, both the arguments and the body of the function are going to be causally sensible. In both cases the DAG structure of the arena is essential, since it allows no enabling from the argument back to the function body. 

\newcommand{\inter}[1]{\,\underset{ #1 }{||}\,}

\newcommand{\intero}{\,{||}\,}

We now have the requisite operations to define the \emph{interaction} and, finally, the \emph{composition} of strategies.
\begin{mydef}[Interaction]
  Given sets of pointer sequences $\sigma \subseteq J_M, \tau\subseteq J_N$ their \emph{interaction} is defined as 
  \[
  \sigma\inter{M,N} \tau = \{ p\in J_{M\cup N} \mid
  p\delm (M\setminus N) \in \tau \land
  p\delm (N\setminus M) \in \sigma \}.
  \]
\end{mydef}

A good intuition for interaction is of two strategies synchronising their actions on the shared moves $M\cap N$.

\paragraph{Observation.}
This definition will be used to compute the interaction of strategies $\sigma:A\Rightarrow B$ and $\tau:B\Rightarrow C$, but we will ignore the issue of tagging of moves as they participate in the definition of composite arenas, and we will just assume the underlying sets of moves are disjoint. This is not technically correct because the arenas can be equal, case in which the tagging is essential to disambiguate the co-products. But the formalisation of tagging, de-tagging and re-tagging is routine and tedious and may obscure the main points. We are sacrificing some formality for clarity.

\begin{myex}
  Let $\tau=\strat(qq'm'(m+1))$ denote a function that increments its argument and $\sigma=\strat(q'q''n''(2n)')$ a function that doubles its argument, $\sigma,\tau:\nata\Rightarrow\nata$. Their interaction $\sigma\intero\tau$ is a set of sequences of the form $qq'q''m''(2m)'(2m+1)$. As written below, the flow of time is top to bottom and each is lined up with the arena in which it occurs. 
  \[
  \begin{array}{ccccc}
    \nata''& \stackrel\sigma\Rightarrow & \nata' &\stackrel\tau\Rightarrow &\nata \\ \hline
           &                            &        &                         & q  \\
           &                            & q       &                         &   \\
    q       &                            &        &                         &   \\
    m       &                            &        &                         &   \\
           &                            & 2m       &                         &   \\
           &                            &        &                         & 2m+1  \\
  \end{array}
  \]
\end{myex}

\begin{myex}\label{ex:iter}
  As defined, a strategy can interact with another strategy only once. Let $\tau=\strat(qq'm'q'n'(m+n))$ denote a function that evaluates its argument twice and returns the sum of received values, and let $\sigma=\strat(q'0')$ be the strategy for constant 0. The interaction $\sigma\intero\tau$ cannot proceed successfully because removing the untagged moves representing the result of $\tau$ leaves sequences of the shape $q'm'q'n'$ which are not in $\sigma$ no matter what values $m,n$ take. 
  \[
  \begin{array}{ccccc}
    I& \stackrel\sigma\Rightarrow & \nata' &\stackrel\tau\Rightarrow &\nata \\ \hline
           &                            &        &                         & q  \\
           &                            & q       &                         &   \\
           &                            & 0       &                         &   \\
           &                            & ?       &                         &   \\
  \end{array}
  \]
\end{myex}

\begin{mydef}[Iteration]
  Given a set of pointer sequences $\sigma\in J_M$ its iteration on $N\subseteq M$ is the set of pointer sequences
  \[
  {!}_N\sigma =\{
  p \in J_{M} \mid \forall ma\cobnd b\ssin p. m\in N \Rightarrow p\selm\{b\}\in\sigma
  \}
  \]
\end{mydef}

A good intuition of iteration is a strategy interleaving its plays. The definition says that if we select moves form an identified subset $N$ and we trace the hereditarily justified plays, they are all in the original set. We can think of each $p\selm \{b\}$ as untangling the `thread of computation' associated with move $ma\cobnd b$ from the interleaved sequence.

\begin{myex}
  Using interaction with the iterated strategy for 0 in Ex.~\ref{ex:iter} is now possible, so ${!}\sigma\intero\tau=qq'0'q'0'0$.
  \[
  \begin{array}{ccccc}
    I& \stackrel\sigma\Rightarrow & \nata' &\stackrel\tau\Rightarrow &\nata \\ \hline
           &                            &        &                         & q  \\
           &                            & q       &                         &   \\
           &                            & 0       &                         &   \\
           &                            & q       &                         &   \\
           &                            & 0       &                         &   \\
           &                            &         &                         & 0  \\
  \end{array}
  \]
  Note that for iteration to have the desired effect it is essential that strategies are equivariant. Consider the situation if the strategy were non-equivariant.
  \[
    \sigma_0 = \{\epsilon, qa\cobnd b, qa\cobnd b\cdot 0b \}
  \]
  where $a,b\in\mathbb A$ are fixed. Iterated pointer sequences such as $qa\cobnd b\cdot 0b\cdot qa\cobnd b\cdot 0b$ are not well formed because the second occurrence of $\cobnd b$ is no longer fresh. However, because the strategy is equivariant we can choose other names when iterating, so that $qa\cobnd b\cdot 0b\cdot qc\cobnd d\cdot 0d$ is legal. 
\end{myex}

Composition is iterated interaction with the synchronisation moves internalised and hidden. 
\begin{mydef}[Composition]
  Given strategies $\sigma : A\Rightarrow B, \tau:B\Rightarrow C$ we defined their \emph{composition} as $\sigma;\tau = ({!}_{I_B}\sigma\inter{M_{A\Rightarrow B},M_{B\Rightarrow C}}\tau )\delm M_B$.
\end{mydef}
We also use $\tau\circ\sigma\stackrel{\mathit{def}}=\sigma;\tau$.

The definition above is the usual \emph{extensional} presentation of strategy composition, which has the slight disadvantage of eliding some of the computational and operational flavour of the games-based approach. An equivalent \emph{intensional} definition can be given using the strategy functions to compute the next move.

\newcommand{\tsin}{\,{\sqsubset}\hspace{-1.34ex}{-}\,}
\begin{mydef}\label{def:stratintens}
  Given strategies $\sigma : A\Rightarrow B, \tau:B\Rightarrow C$ we define their \emph{interaction function} as $\reallywidehat{\sigma\lightning\tau}: \Play_{(A\Rightarrow B)\Rightarrow C} \rightarrow \mathcal P(M_{(A\Rightarrow B)\Rightarrow C}\times {\mathbb {A}}^2)$,
  \[
  \reallywidehat{\sigma\lightning\tau}(p) =
  \hat\tau(p \downharpoonright M_A) \cup
  \bigcup_{\substack{q a \cobnd b \tsin p \\ q\in I_B}}\hat\sigma(p \upharpoonright qa \cobnd b)
  \]
\end{mydef}

As in the case of the extensional definition, the definition is asymmetrical. Unlike the extensional definition the intensional definition makes some features of composition clearer. The first one is that the behaviour of the composite strategy in the second component ($B\Rightarrow C$) only depends on the history of the play as restricted to that component, $p\downharpoonright M_A$. In other words, the strategy $\tau$ does not `see' what $\sigma$ is up to.
\begin{center}
\includegraphics[scale=1.2]{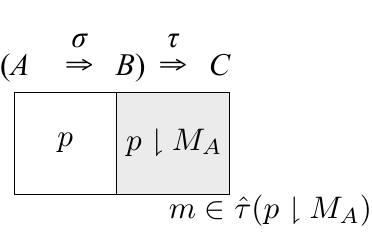}
\end{center} 
The second one is that the behaviour of the composite strategy in the first component ($A\Rightarrow B$) is restricted not only to just the history of the play in that component, but also to each `thread' of the strategy $p\upharpoonright q$, going back to some initial move $q\in I_B=I_{A\Rightarrow B}$.
\begin{center}
  \includegraphics[scale=1.2]{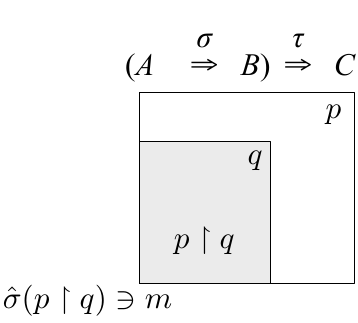}
\end{center}
This definition also makes it more apparent that there is an implicit nondeterminism when a move occurs in the shared arena $B$, as both $\sigma$ and $\tau$ can continue playing independently. We will see in Sec.~\ref{sec:identity} that this has some important consequences.

To arrive at composition itself we need to hide the moves in the interface arena $B$ from each $\reallywidehat{\sigma\lightning\tau}(p)$.
This concludes the detour into the intensional presentation of strategy composition.
\\[2ex]
\begin{mythm}\label{thm:comp}
Given strategies $\sigma : A\Rightarrow B, \tau:B\Rightarrow C$ their composition is also a strategy  $\sigma; \tau : A\Rightarrow C$.
\end{mythm}
\begin{proof}
  Equivariance is preserved by all operations above, from the general principles of the theory of nominal sets~\citep{pitts2013nominal}. 

  Prefix-closure and $O$-closure are immediate by unwinding the definitions.

  It remains to show that the sequences are valid plays of $A\Rightarrow C$. We already know the pointer sequences are well formed. 
  The first step is to show that ${!}_{I_B}\sigma\inter{M_{A\Rightarrow B},M_{B\Rightarrow C}}\tau$ is a strategy in arena $(A\Rightarrow B)\Rightarrow C$, which is immediate from definitions.
  The second step is to show that $\delm M_B$ gives plays in $A\Rightarrow C$, which is true because pointers from $C$-moves to $B$-moves to $A$-moves are replaced by pointers from $C$-moves directly to $A$-moves. 
\end{proof}

\section{Game semantic models}

\subsection{Cartesian closed categories}

In this essay we are focussing on programming languages that build on the (call-by-name) lambda calculus, so we will focus on games which can model it. Instead of relating directly the syntactic and the semantic models, it is standard to use an abstract mathematical model expressed in terms of \emph{category theory}.\footnote{For the current presentation a minimal familiarity with the basic concepts of this topic is required; accessible introductions and tutorials abound (e.g.~\citep{samsont,milewski}).} For our language, this model is known as a \emph{Cartesian closed category}~\citep{lambek1988introduction}. It is an important model, beautiful in its simplicity, deeply connected with type-theoretical and logical aspects of computations. \cite{baezs} give a fascinating account of these connections.

We will start by attempting to identify a category of games where objects are arenas $A, B, \ldots$ and morphisms $\sigma : A \to B$ are strategies $\sigma : A \Rightarrow B$. The product of two objects is the arena product $A\times B$, the terminal object is the empty arena $I$, and the exponential object is the arrow arena, $B^A= A\Rightarrow B$.

Composition is well defined (Prop.~\ref{thm:comp}), but we still need to verify its required categorical properties: \emph{associativity} and the existence of an \emph{identity} strategy for composition for all arenas. 

\subsubsection{Associativity}

We show the proof of composition in detail, for didactic reasons. The justification pointers can be an awkward mathematical structure, and the original formulation, based on numerical indices into the sequence is particularly unfortunate since any operation on indices requires re-indexing. As a consequence, the original proof of associativity is rather informal. The use of \emph{names} as a representation for pointers eliminates the need for re-indexing and can allow proofs that are both rigorous and elementary. The only challenge of the proof lies in the careful unpacking of several layers of complicated definitions, but once this bureaucracy is dealt with the reasoning is obvious. Other proofs in this presentation are similarly elementary, if tedious. Avoiding these complications can be achieved, but at the cost of some significant additional mathematical sophistication~\citep{lmcs:3966}. 

\begin{mythm}[Associativity]
For any three strategies $\sigma:A\Rightarrow B, \tau:B\Rightarrow C, \upsilon : C\Rightarrow D$, $(\sigma;\tau);\upsilon=\sigma;(\tau;\upsilon)$. 
\end{mythm}
\begin{proof}
  Elaborating the definitions, the LHS is
  \begin{equation}
    \label{eq:ass1}
     (!_C((!_B\sigma\inter{AB,BC}\delm B)\inter{ABC,CD}(\nu\delm B))\delm C 
     \quad = (!_C((!_B\sigma\inter{AB,BC}\tau)\inter{ABC,CS}\nu)\delm BC 
  \end{equation}
  There are no $B$-moves in $ (!_C((!_B\sigma\inter{AB,BC}\delm B)$, so we can extend the scope of $\delm B$.

  Elaborating the definitions, the RHS is
  \begin{align}
    &(!_B\sigma\inter{AB,BD}(!_C\tau\inter{BC,CD}\nu)\delm C)\delm B  \\
    & \quad = (!_B(\sigma\delm C)\inter{AB,BD}(!_C\tau\inter{BC,CD}\nu)\delm C)\delm B \label{eq:ass2} \\
    & \quad = (!_B\sigma\inter{AB,BD}(!_C\tau\inter{BC,CD}\nu)\delm C)\delm BC \label{eq:ass3} 
  \end{align}
  Eq.~\ref{eq:ass2} is true because there are no C-moves in $\sigma$, so $\sigma\delm C= \sigma$.

  Eqn.~\ref{eq:ass3} is true because $\delm C$ distributes over concatenation.

  Therefore, it is sufficient to show the expressions in Eqns.~\ref{eq:ass1} and~\ref{eq:ass2} are equal.

  Let $p\in {!}_C((!_B\sigma\inter{AB,BC}\tau)\inter{ABC,CS}\nu$ is equivalent, by definition with
  \begin{align}
    & p\delm D\in {!_C}(!_B\sigma\inter{AB,BC}\tau)     \label{eq:ass4} \\
    \land\ & p\delm AB\in \nu \label{eq:ass5}
  \end{align}
  By elaborating the definitions:
  \begin{align}
    \text{Prop.~\ref{eq:ass4} }
    \Leftrightarrow \forall m a\cobnd b.m\in I_C \Rightarrow &\  p\delm D\selm\{b\}\in {!_B}\sigma\inter{AB,BC}\tau
    \label{eq:ass6} \\
    \Leftrightarrow\ & p\delm D\selm\{b\} \delm C \in {!_B}\sigma
    \label{eq:ass7} \\
    & \land  p\delm D\selm\{b\} \delm A \in \tau \Leftrightarrow p\selm\{b\}\delm A\in\tau
    \label{eq:ass8}
  \end{align}
  The equivalence in Eqn.~\ref{eq:ass8} holds because once we restrict to moves hereditarily justified by a C-move ($m\in I_C$), the removal of D-moves has no effect since the hereditarily justified play is restricted to arenas ABC.

  Elaborating the definition yet again,
  \begin{align}
    \text{Prop.~\ref{eq:ass7}} & \Leftrightarrow
    \forall n c\cobnd d \ssin p\delm D\selm\{b\}\delm C. n\in I_B \\
    & \Rightarrow
    p\delm D\selm \{b\}\delm C\selm\{d\} \in \sigma \Leftrightarrow
    p\selm\{d\} \in \sigma
  \end{align}
  Because restricting to the hereditarily justified play of a B-move ($ n\in I_B $) makes the other restrictions irrelevant.

  To summarise,
  \begin{multline} \label{eq:ass9}
    p\in {!_C}((!_B\sigma\inter{AB,BC}\tau)\inter{ABC,CS}\nu
    \Leftrightarrow \\
    \forall m a\cobnd b.m\in I_C \Rightarrow p\selm\{b\} \in \tau \land
    \forall n c\cobnd d \ssin p\selm\{b\}. n\in I_B \Rightarrow  p\selm\{d\} \in \sigma,
  \end{multline}
%  which makes good intuitive sense: in any play over the interaction sequence of the three strategies, if we hereditarily justify from an initial move in arena $C$ we get a $\tau$-play and if we do the same for a B-move we get a $\sigma$-play. 

  This was the more difficult case.

  $p\in {!}_B\sigma\inter{AB,BD}(!_C\tau\inter{BC,CD}\nu)$ is equivalent to the same conditions as in Eqn.~\ref{eq:ass9} simply by elaborating the definitions, except that Prop.~\ref{eq:ass8} appears as $ p\delm A\selm\{b\} \in\tau $, but $\selm, \delm$ commute in this case. 
\end{proof}

Composition is not only associative but also monotonic with respect to the inclusion ordering:
\begin{mythm}[Monotonicity]\label{thm:mot}
  If $\sigma\subseteq \sigma'$ then for any $\tau, \upsilon$, $\sigma;\tau\subseteq\sigma';\tau$ and $\upsilon;\sigma\subseteq\upsilon;\sigma'$.
\end{mythm}
The proof is immediate, since all operations involved are monotonic.

\subsubsection{Identity}\label{sec:identity}

The second challenge is the formulation of an appropriate notion of \emph{identity} strategy for any arena. A candidate for identity $\kappa_A:A_0\Rightarrow A_1$ is a strategy which immediately replicates $O$-moves from $A_0$ to $A_1$ and vice versa, while preserving the pointer structures -- a so-called \emph{copy-cat strategy}. 
\begin{mydef}[Copy-cat]
  For any arena $A$ we defined $\kappa_A$ as
%  a set of justified sequences over arena $A_0\Rightarrow A_1$, the smallest such as
%\begin{align*}
%  &q\in I_A \Rightarrow q_1a\cobnd b \cdot q_0 b\cobnd c \in \kappa_A\\
%  &p\in\kappa_A,  m_ia\cobnd b \cdot m_{1-i}c\cobnd d\ssin p 
%  \land m\in O_{A_0\Rightarrow A_1}  \Rightarrow \\
%  &\qquad n\in Q_A\Rightarrow p\cdot n_jd\cobnd e\cdot n_{1-j}b\cobnd f \in \kappa_A \\
%  &\qquad \land n\in A_A \Rightarrow p\cdot n_jd \cdot n_{1-j}b \in \kappa_A.
%\end{align*}
  \begin{align*}
    \hat\kappa_A (q_1a\cobnd b) &= q_0 b\cobnd c \\
    \hat\kappa_A \left(p \cdot m_ia\cobnd b \cdot m_{j}c\cobnd d\cdot p'\cdot n_jd\cobnd e\right) &= n_{i}b\cobnd f 
%    \hat\kappa_A \left(p \cdot m_ia\cobnd b \cdot m_{j}c\cobnd d\cdot p'\cdot n_jd\right) &= n_{i}b. 
  \end{align*}
  where $i\neq j\in \{0,1\}$.
\end{mydef}
Graphically, the strategy can be represented informally as:
\begin{center}
  \includegraphics[scale=1.2]{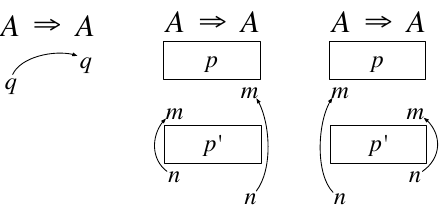}
\end{center}
where copies of the same arena ($A$), initial question ($q$) or move ($m,n$) are indicated by using the same variable. 

A copy-cat strategy has `the same behaviour' in both components:
\begin{mylem}
  $\kappa_A\delm A_0 = \kappa_A\delm A_1$, up to a relabeling of moves. 
\end{mylem}
\begin{mythm}
$\kappa_A : A_0\Rightarrow A_1$ is a strategy. 
\end{mythm}

However, $\kappa_A $ is, perhaps surprisingly, not a unit for composition. Consider for example $A=unit$ and the strategy $\rho:unit''\Rightarrow unit'=\{q'q''a'a''\}$. This strategy, whereby the argument acknowledges termination after the body of the function is akin to a process-forking call. The interaction of $\rho;\kappa$ is shown below:

\begin{center}
  \includegraphics[scale=1.2]{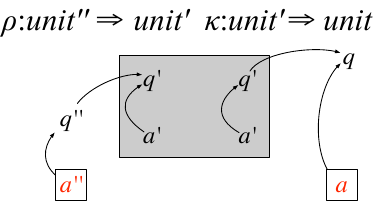}
\end{center}
 
The synchronisation in the shared arena and the ordering of moves in $\kappa_\mathit{unit}$ impose no particular order between the moves highlighted in the diagram  ($a, a''$), because of the inherent nondeterminism of strategy composition. Therefore the play $qq''a''a$ is in the composition but absent from the original strategy $\rho$.

However, even though $\kappa_A$ is not in general an identity for composition, it is \emph{idempotent}, i.e. it is an identity when composed with itself. 
\begin{mylem}
The $\kappa_A$ strategy is \emph{idempotent}, that is for any arena $A$ we have that $\kappa_A=\kappa_A;\kappa_A.$
\end{mylem}
The proof is similar to that of Prop.~\ref{lem:synch}.

Identifying an idempotent morphism means that we can now construct a proper category using the so-called \emph{Karoubi envelope} construction~\citep{BALMER2001819}.

\begin{myprop}
  There exists a category of games in which objects are arenas $A$, identities at $A$ are \emph{copy-cats} $\kappa_A$ and morphisms $\sigma^\dag:A\to B$ are \emph{saturated strategies} $\sigma^\dag=\kappa_A;\sigma;\kappa_B$, where $\sigma:A\Rightarrow B$ is a strategy. 
\end{myprop}

We call strategies $\sigma^\dag$ \emph{saturated} because through composition with the copy-cat strategies $\kappa$ new behaviours are added. We will discuss this further when we talk about definability (Sec.~\ref{sec:definability}).

\subsubsection{Cartesian closed structure}

To be able to model at least call-by-name lambda calculus the category above needs to be Cartesian closed. And, indeed,

\begin{mythm}[CCC]\label{thm:ccc}
  The category of games and saturated strategies is Cartesian closed:
  \begin{description}
  \item[Terminal object] is the arena $I$ with no moves $M_I=\emptyset$
  \item[The product] of two arenas $A_1, A_0$ is the arena $A_1\times A_0$ with projections $\pi_i:A_1\times A_0\rightarrow A_i'$, 
    $\pi_i = \kappa_{A_i}.$
  \item[The exponential] of two arenas $A,B$ is the arena $A\Rightarrow B$ with
    \begin{description}
    \item[Evaluation morphism] $ev_{A,B}:(A_0\Rightarrow B_0)\times A_1\rightarrow B_1$ with $ev_{A,B}=\varepsilon_{A,B}^\dag$ where
      \begin{align*}
        \varepsilon_{A,B} (pm) = \hat\kappa_A(pm) \text{ if } m\in A_i\\
        \varepsilon_{A,B} (pm) = \hat\kappa_B(pm) \text{ if } m\in B_i.
      \end{align*}
    \item[Transpose] of any strategy $\sigma:A\times B\rightarrow C$ is the strategy $\lambda\sigma:A\rightarrow (B\Rightarrow C)$ where $\lambda$ is a re-tagging of moves. 
    \end{description}
  \end{description}
\end{mythm}
\begin{proof}
  \begin{description}
  \item[Terminal object:] The only strategy $!:A\rightarrow I$ contains the empty play since no move in $A$ is enabled.
  \item[Product:] For every object B and pair of morphisms $\sigma_i^\dag : B \to A_i$ the product of the two morphisms $\langle \sigma^\dag_1, \sigma^\dag_0\rangle=\sigma^\dag_1\cup\sigma^\dag_0$. The fact that $\sigma^\dag_i=\langle\sigma^\dag_1,\sigma^\dag_0 \rangle;\pi_i$ is immediate because $\pi_i$ is essentially a (saturated) copy-cat which preserves $\sigma_i$ and removes $\sigma_{1-i}$ as none its moves are enabled.
  \item[Exponential:] Evaluation is a combination of copy-cat strategies, on the $A$ and $B$ components, respectively. The \emph{re-tagging} defining $\lambda$ is induced by the two isomorphic ways in which the coproduct can associate. We leave the details as an exercise to the reader. 
  \end{description}
\end{proof}

\subsection{Interpreting PCF}

Prop.~\ref{thm:ccc} along with Prop.~\ref{thm:mot} show that the category of saturated games is a model for call-by-name lambda calculus with recursion. For the sake of simplicity we will leave recursion aside and concentrate on the recursion-free language. The interpretation is the standard one.

Let $\theta$ stand for the types of PCF. Let $\beta$ be the base-types of the language (naturals, booleans, unit, etc.) and let $\theta\rightarrow\theta$ be the only type-forming construct. The constants of the language are the base-type constants together with base-type operations (arithmetic and logic) and if-then-else. The other term formers are lambda-abstraction ($\lambda x.t$) and application ($t\, t'$). Let us write $\mathit{fv}(t)$ for the free variables of a term $t$, defined as usual. 

We use typing judgments of the form $\Gamma\vdash t:\theta$, with $\Gamma=x_0:\theta_0,\ldots,x_n:\theta_n$ a set of variable type assignments, $\mathit{fv}(t)\subseteq \mathrm{dom}(\Gamma)$. The judgment is read as ``if variables $x_i$ have types $\theta_i$ as given by $\Gamma$, then $t$ has type $\theta$.'' The judgments are checked using the following rules, expressed in natural deduction style:

\begin{center}
  \centerAlignProof
  \AxiomC{}
  \UnaryInfC{$\Gamma,x:\theta\vdash x:\theta$}
  \DisplayProof\quad
  \AxiomC{$\Gamma\vdash t':\theta\rightarrow\theta'$}
  \AxiomC{$\Gamma\vdash t:\theta$}
  \BinaryInfC{$\Gamma\vdash t'\ t:\theta'$}
  \DisplayProof\quad
  \AxiomC{$x:\theta,\Gamma\vdash t':\theta'$}
  \UnaryInfC{$\Gamma\vdash \lambda x.t':\theta\rightarrow \theta'$}
  \DisplayProof
\end{center}

\newcommand{\sbr}[1]{\llbracket{#1}\rrbracket}

The interpretation function is written as $\sbr-$. Types are interpreted as arenas:
\[
  \sbr{bool}= \mathit{Bool}, etc. \qquad
  \sbr{\theta\rightarrow\theta'}=\sbr\theta\Rightarrow\sbr{\theta'}.
\]
Variable type assignments are interpreted as products:
\[
\sbr{\emptyset}=I,\qquad
\sbr{x:\theta, \Gamma}=\sbr\theta\times\sbr\Gamma.
\]
Terms $\Gamma\vdash t:\theta$ are interpreted as strategies on arena $\sbr\Gamma\Rightarrow\sbr\theta$, inductively on the (unique) derivation of the type judgment. The interpretation of the constants as strategies was already given in the preceding sections. Variables, abstraction and application are interpreted canonically using the categorical recipe:
\begin{align*}
  \sbr{x:\theta,\Gamma\vdash x:\theta}&=\pi_0 :\sbr\theta\times\sbr\Gamma\rightarrow \sbr\theta\\
  \sbr{\Gamma\vdash t'\ t:\theta'}&=
  \langle \sbr{\Gamma\vdash t':\theta\rightarrow\theta'},
  \sbr{\Gamma\vdash t:\theta}\rangle;ev_{\sbr\theta,\sbr{\theta'}}\\
  \sbr{\Gamma\vdash \lambda x.t':\theta\rightarrow \theta'}&=
  \lambda\sbr{x:\theta,\Gamma\vdash t':\theta'}.
\end{align*}

\section{Definability}\label{sec:definability}

The saturated unrestricted model described here contains many behaviours which are not syntactically definable in PCF. A simple example would be the non-deterministic coin-flip strategy ($\sigma_{\mathit{flip}}=\strat\{q\cdot\mathit{true}, q\cdot\mathit{false}\}$). In this section we will determine a syntactic extension for PCF which restores definability. One may think of it as an `axiomatisation' of the game-semantic model. 

The saturation of strategies might be worrying since it involves a loss of control over the order in which certain moves occur. Can we still have languages with sequencing or synchronisation? The property below gives a positive answer.
\begin{mylem}[Synchronisation]\label{lem:synch}
  Let $\sigma=\strat\{p\cdot ma \cobnd b\cdot n c \cobnd d\cdot p'\} : A\Rightarrow B$. Then,
  $p\cdot n c \cobnd d\cdot m a \cobnd b\cdot p'\in\sigma^\dag$ if and only if $m\in P_{A\Rightarrow B}$ or $n\in O_{A\Rightarrow B}$.
  This also holds if $a=c$.
\end{mylem}
\begin{proof}
  We need to consider all combinations for $m$ and $n$ to be $O$ or $P$ moves in $A\Rightarrow B$, and also whether they occur in $A$ or in $B$. Because $\kappa$ always copy-cats $O$ moves to $P$ moves, and because in this case the $P$ move occurs after the $O$ move ultimately it does not matter whether $\kappa_A$ or $\kappa_B$ does the copying, so for this argument it does not matter whether $m, n$ occur in $A$ or $B$. Let us assume they occur in $B$
  \begin{enumerate}
  \item $m\in O_B\subseteq O_{A\Rightarrow B}$ and $n\in P_B\subseteq P_{A\Rightarrow B}$. When composing $\sigma$ with $\kappa_{B}$ the polarities of the moves in arena $B\Rightarrow B'$ are reversed, $m\in P_{B\Rightarrow B'}$, $n\in O_{B\Rightarrow B'}$. As usual, we use $B, B'$ to distinguish the two occurrences of arena $B$ in the composite arena $B\Rightarrow B'$. This means that $m$ is necessarily the copy-image of a $B'$ move occurring earlier and $n$ will be copied into a $B'$ move occurring later. Following the hiding of the arena $B$ the order of the moves $m$ and $n$ must necessarily stay the same.
  \item $m\in O_B\subseteq O_{A\Rightarrow B}$ and $n\in O_B\subseteq P_{A\Rightarrow B}$. When composing $\sigma$ with $\kappa_{B}$ the polarities of the moves in arena $B\Rightarrow B'$ are reversed, $m\in P_{B\Rightarrow B'}$, $n\in P_{B\Rightarrow B'}$. This means that $m$ is necessarily the copy-image of a $B'$ move occurring earlier and so is $n$. Both these moves are $O$-moves in arena $B\Rightarrow B'$, and they may occur in any order since $\kappa_B$ accepts both orders. So after hiding the arena $B$ in composition both orders of $m$, $n$ may be present in the composite strategy.
  \item All other cases are similar to the previous case (2).
  \end{enumerate}
\end{proof}

This Proposition gives in fact a rational reconstruction of the \emph{permutative saturation} condition used in asynchronous game semantics used by~\cite{DBLP:conf/concur/Laird05,DBLP:journals/apal/GhicaM08} and more broadly in semantics of asynchronous communication~\citep{Udding1986}. The reason that we can permute all sequences of moves $m\cdot n$ in which $m$ does not justify $n$, unless $m$ is an $O$-move and $n$ a $P$-move is, intuitively, that $P$ should always be able to synchronise on $O$. This is reflected, mathematically, by the fact that the saturated copy-cat $\kappa$ will always copy $O$-moves into $P$-moves. The Proposition is also useful because it will allow us to follow quite closely the definability procedure used by~\cite{DBLP:journals/apal/GhicaM08}.

It is important to emphasise at this stage that the simple game model we use here takes an `\emph{angelic}' perspective on termination. We mentioned earlier the $\sigma_{\mathit{flip}}:\mathit{Bool}$ non-deterministic strategy. Let us consider a `non-deterministic projection' strategy $\sigma_\pi:\mathit{Unit\times Unit'\to Unit''}$, $\sigma_\pi=\strat\{q''\cdot q\cdot a\cdot a'', q''\cdot q'\cdot a'\cdot a'' \}$, along with the  `\emph{diverging}' unit strategy $\sigma_\Omega=\{q\}:\mathit{Unit}$ and the non-diverging strategy at the same type $\sigma_{\mathit{nop}}=\strat\{q\cdot a\}:\mathit{Unit}$. It follows that $\langle \sigma_\Omega, \sigma_{nop}\rangle;\sigma_\pi = \sigma_{nop}$. In words, the non-deterministic choice between a responsive and a non-responsive strategy will always be the responsive strategy, hence the `angelic' moniker. More precise models can include separately the non-responsive plays, called \emph{divergences}, similar to the trace models of~\cite{DBLP:journals/jacm/BrookesHR84}, as adapted to game semantics by~\cite{DBLP:conf/lics/HarmerM99}.

Another important feature of a game model is whether it is \emph{extensional}, like that of \cite{DBLP:journals/entcs/AbramskyM96} or \emph{intensional}, like that of \cite{DBLP:journals/iandc/HylandO00}. Two strategies $\sigma_1, \sigma_2:I\to A$ are said to be \emph{equivalent} $\sigma_1\equiv \sigma_2$ if for all \emph{test strategies} $\tau:A\to \mathit{Unit}$, $\sigma_1;\tau=\sigma_2;\tau$, i.e.  $\sigma_1;\tau=\sigma_2;\tau=\sigma_\Omega$ or $\sigma_1;\tau=\sigma_2;\tau=\sigma_{nop}$.
\begin{mylem}
  Two strategies $\sigma_1, \sigma_2:I\to A$ are \emph{equivalent} $\sigma_1\equiv \sigma_2$ if and only if their saturations are equal $\sigma_1^\dag=\sigma_2^\dag$.
\end{mylem}
\begin{proof}
Equal strategies are equivalent, obviously. If two strategies are not equal then there is a play $p$ in their symmetric difference. Composition with $\strat\{q\cdot p\cdot a\}: A\to \mathit{Unit}$ will yield $\sigma_\Omega$ for the strategy not containing $p$ in its saturation, and $\sigma_{nop}$ for the other one. 
\end{proof}
This makes the model extensional.

Note that unlike IA or concurrent IA (ICA)~\citep{DBLP:journals/apal/GhicaM08} saturated strategies are not characterised by their `set of complete plays', i.e. those plays in which the initial question is answer. This is because strategies may contain plays such as $q\cdot a\cdot q'\cdot a'$ in which moves happen after the initial move was answered. Unlike intensionality, which makes the model harder to use directly, the fact that strategies are not characterised by complete plays is not problematic from a technical point of view.

We are now ready to address the question of definability: what syntax do we need so that any strategy, or rather its saturated version, is the denotation of some term. We will follow the definability procedure of ICA, since the two models are similar enough. The ingredients are:
\begin{description}
\item[State.] ICA has \emph{local variables}, which are used to record and test the order of execution of moves by associating each move either with writing to the state, or reading from the state. 
\item[Semaphores.] ICA also has \emph{local split binary semaphores} to achieve synchronisation between sub-plays, which can be seen as `\emph{threads}'. The type of semaphores is isomorphic to $\mathit{Unit_1\times Unit_2}$ and have strategy $\sigma_{sem}:\mathit{(Sem \Rightarrow Unit)\Rightarrow Unit'}$, $\sigma_{sem}=\strat\{q'q(q_1a_1q_2a_2)^*aa'\}$, strictly alternating between the two components of the semaphore type. Intuitively, one represents a `\emph{grab}' action, and the other a `\emph{release}'.
\item[Concurrency.] ICA has a \emph{static parallelism} strategy $\sigma_{par}:\mathit{Unit'\Rightarrow Unit''\Rightarrow Unit}$, $\sigma_{par}=\strat\{qq'q''a'a''a\}$. It interleaves its arguments in any order but it only terminates when both arguments terminate. This gives a `\emph{fork/join}' structure to all strategies, in the sense that only unanswered questions may justify, and they may only be answered when all the questions they justify have been themselves answered. This is the construct we drop, and we replace it with a simpler, dynamic concurrency strategy $\sigma_{run}:\mathit{Unit'\Rightarrow Unit}$ which allow its argument to finish before or \emph{after} the initial question is answered. In other words, $\sigma_{run}=\strat\{qq'aa'\}$. 
\end{description}

With this rather minor change, the definability procedure for ICA can be replicated, giving us a complete syntax for the model. The only possibly significant distinction between the two models is in the fact that ICA strategies are characterised by complete plays, which means that a play can be safely removed from a strategy if the initial question is not answered. This could have a potential impact on definability because we could, for example, allow several $P$-moves to occur and, subsequently, if we decide the play is not proceeding in a desired play we can simply stop playing in it. Such an artifice would not work in this model. However, the definability argument for ICA does not make use of it. 

For instance, the strategy for \emph{catch} in Sec.~\ref{sec:control} can be defined in this model, but not in the ICA model because in the case the exception is raised it breaks the fork/join discipline by providing the answer to the initial question before pending questions have been answered. However, a word of caution. Operationally, the definability argument and therefore the reconstruction of \emph{catch} creates a large number of concurrent threads, killing off those plays that evolve in an undesirable direction by introducing divergences, which are subsequently hidden by the angelic notion of observation. This reconstruction of the strategy is artificial, of little practical importance.

\section{Conclusion}

This presentation is meant primarily as a didactic exercise, presenting a game model where the plays are governed by no combinatorial restrictions such as bracketing, alternation, innocence, etc. The structure we preserve is the proper justification of plays, which we take it to have essentially a causal rationale. We also preserve certain closure conditions of strategies which also have causal or temporal motivations, while adding a new one via the Karoubi envelope construction: permutative asynchronous saturation. Behind this construction lies a technical reason, the fact that the `intuitive' copy-cat strategy does not do the job of being a unit for composition with arbitrary strategies. However, saturation is consistent with similar closure conditions tracing back to the early literature on asynchronous concurrency. Of course, saturation could have been stipulated as a required closure property of strategies, as~\cite{DBLP:journals/apal/GhicaM08} do. The choice to present it as an `add-on' property is also didactically motivated. The direct intuition behind saturation is, arguably, not as compelling as that behind, say, prefix-closure or equivariance, so there could be a risk of baffling the newcomer to game semantics. But there is also some deliberate methodological candor in showing how the game-semantic sausage is made. Game semantics is sometimes fiddly. Its power is derived from its ability to model the sometimes messy reality of programming languages. Saturation is an example of a condition that can be reconstructed by backtracking from mathematical considerations.

Methodologically, this semantics harkens back to the early days of semantic programming languages when the denotational approach was deemed sufficient. In time, operational semantics became the dominant specification paradigm, with denotational models seen more like characterisations of observational equivalence, rather than prima facie specifications. And, indeed, some of the more abstract denotational models, such as functor-category models of state~\citep{DBLP:journals/tcs/OHearnPTT99}, are far removed from the mechanics of evaluation. In contrast, game models, especially when strategies are presented via a next-move function, are close enough to evaluation to allow for quantitative modeling~\citep{DBLP:conf/popl/Ghica05} or semantics-directed compilation~\citep{DBLP:conf/popl/Ghica07}.
In fact operational semantics and game semantics are eminently compatible, since the latter can be used to give a compositional formulation to the former via so-called `\emph{system-level game semantics}'~\citep{DBLP:journals/entcs/GhicaT12}. 

Finally, Abramsky's idea of a \emph{semantic cube} is compelling and inspirational beyond definability, and beyond game semantics. In semantics of programming languages the idea of a taxonomy of behaviours is usually automatically set in the context of types. But types are largely related to the input-output (extensional) shape of computation, whereas the features that form the dimensions of the Abramsky cube concern the intensional shape of computation, and it cuts across the type discipline. Indeed, all languages we have considered in this essay, extensions of PCF, satisfy the simply-typed discipline. This deeper semantic taxonomy which we associate with the Abramsky Cube has direct relevance on equational reasoning, as pointed out for examble by~\cite{DBLP:journals/jfp/DreyerNB12}. More recent work by \cite{DBLP:journals/corr/abs-1907-01257}, which presents a model of computation based on hyper-graphs, shows how these metaphorical shapes of computation can be realised as actual shapes of graphs arising in the process of computation. These all suggest that the ideas behind Abramsky's Cube are likely to exert a lasting influence on the study of programming languages.

\end{document}